%% file: main_LCSS_final.tex
\documentclass[letterpaper, 10 pt, conference]{ieeeconf}  

\IEEEoverridecommandlockouts                              

\usepackage{amsmath} 
\usepackage{amssymb}  

\usepackage{amssymb}
\usepackage{amsmath}
\usepackage{color}
\usepackage{graphicx}
\usepackage{dsfont}
\usepackage{mathrsfs}
\usepackage{etoolbox}
\usepackage{float}
\usepackage{subfigure}
\usepackage{amsbsy}
\usepackage{mathabx}
\usepackage{comment}
\usepackage{hyperref}
\usepackage{algorithm}
\usepackage{algpseudocode}
\usepackage{caption}
\usepackage{relsize}
\usepackage{siunitx}
\usepackage{soul}
\usepackage[noadjust]{cite}
\usepackage{textcomp}
\usepackage{multirow}
\usepackage[normalem]{ulem}

\definecolor{modcol}{RGB}{255,102,178}
\newtheorem{SERMprob}{\textbf{Problem}}

\definecolor{MFabg}{RGB}{76, 191, 229}

\newtheorem{thm}{Theorem}
\newtheorem{lem}{Lemma}

\newcommand{\colvec}[2][.9]{%
	\scalebox{#1}{%
		\renewcommand{\arraystretch}{.9}%
		$\begin{bmatrix}#2\end{bmatrix}$%
	}
}

\input{shortcuts_math}

\title{A General Regularized Distributed Solution for\\ System State Estimation from Relative Measurements}

\author{Marco~Fabris, Giulia~Michieletto~\IEEEmembership{Member,~IEEE,} and~Angelo~Cenedese~\IEEEmembership{Senior Member,~IEEE}	
	\thanks{M. Fabris is with the Dept. of Aerospace and Control Engineering, Technion, Haifa, Israel.
		G. Michieletto is with the Dept. of Management and Engineering, A. Cenedese is with the Dept. of Information Engineering, both at University of Padova, Padova, Italy.}
	\thanks{\mbox{Corresp. author: M. Fabris, {\tt \scriptsize \href{mailto:marco.fabris@campus.technion.ac.il}{marco.fabris@campus.technion.ac.il}}}
	}
	\thanks{Part of this work was supported by MIUR (Italian Ministry for Education)
		under the initiative ``Departments of Excellence" (Law 232/2016).}
}

\begin{document}

	\maketitle
	\thispagestyle{empty}
	\pagestyle{empty}

	\begin{abstract}
	
This work presents a novel general regularized distributed solution for the state estimation problem in networked systems. 
		Resting on the graph-based representation of sensor networks and adopting a multivariate least-squares approach, the designed solution exploits the set of the available inter-sensor relative measurements and leverages a general regularization framework, whose parameter selection is shown to control the estimation procedure convergence performance. 
As confirmed by the numerical results, this new estimation scheme allows $(i)$ the extension of other approaches investigated in the literature and $(ii)$ the convergence optimization in correspondence to any (undirected) graph modeling the given sensor network. 
		
	\end{abstract}


	
	\section{Introduction} 
	\label{sec:intro}
	
A \emph{Sensor Network} (SN) can be generically defined as a system composed of multiple devices (sensor nodes) endowed with limited computation, sensing  and communication capabilities, which interact in order to solve problems that are beyond the capacity and knowledge of each single element, accomplishing complex global tasks through the realization of simple local rules. Characterized by a cooperative nature, in the latest years, {SNs} have become a very popular enabling technology in several cutting-edge research areas, including, e.g., 	Internet-of-Things~\cite{guastella2020cooperative}, and Smart Sensing~\cite{lissandrini2019cooperative}.  

A canonical as well as fundamental problem within this context consists in the \emph{state estimation from relative measurements} (SERM), namely in the determination of the state of each sensor node composing the network, generally resting on the exploitation of a set of noisy relative measurements through a distributed approach~\cite{barooah2007estimation}. This issue emerges as a general problem in many diverse SN applications, as, for instance, clock synchronization in wireless SNs, self-calibration of visual SNs, power state estimation in smart grids (see, e.g.,~\cite{xiong2017cooperative,simonjan2019decentralized,dehghanpour2018survey} and the references therein).

\medskip 
\noindent{\textbf{Related works} -} 
Motivated by the pervasiveness of this task, many algorithms have been proposed over the years. 
These mainly rely on filtering strategies (particle or extended Kalman filters), on maximum a posteriori or maximum likelihood estimation criteria, on the resolution of least-squares (LS) problems basing on consensus-based agreement protocols~(see, e.g.,~\cite{trimpe2014event,farina2016distributed,aster2018parameter} and the references therein).

{In particular, this latter optimization-based method allows to investigate the {SERM} algorithm performance in connection with the topological features of the SN, especially in relation to the quality and quantity of the available measurements.} 
{In this direction, for instance, the convergence properties of a LS-based {SERM} algorithm are examined in~\cite{ravazzi2018distributed} given that the relative measurements are characterized by heterogeneous and uncertain quality. 
Similarly, in~\cite{shi2020bias}, the estimation performance is studied in the light of the amount of the available measurements, also corrupted by a constant bias; different behavior are identified for bipartite and nonbipartite tolopogies. 
This distinction emerges also in~\cite{fabris2019distributed}, where a distributed parametric iterative SERM strategy is presented, and such method is similar to the solution proposed in~\cite{fabris2020proximal}: both these algorithms involve a penalty parameter that acts as a regularization coefficient and allows to improve the convergence performance by prioritizing information sharing through certain network links. }  
	
\medskip
	\noindent{\textbf{Contributions} -} 
 {Inspired by the strategies in}~\cite{fabris2019distributed,fabris2020proximal}, this work provides a comprehensive point on the SERM problem and, then, formalizes a \textit{general regularized distributed solution} (GRDS) by broadening the groundwork already established in the field.
{Indeed, both the consensus-like iterative schemes $\Sigma_{\eta}$ and $\Sigma_{\rho}$ devised in~\cite{fabris2019distributed} and~\cite{fabris2020proximal}, respectively, can be revisited as specific cases of the here proposed GRDS framework. 
In addition, this holds also for the well-known consensus-based scheme $\Sigma_{\epsilon}$ detailed in \cite{olfati2007consensus}. 
Finally, such reformulation allows the procedure performance optimization for any kind of connected undirected network.}

Specifically, the purpose of this work is 
$(i)$ to design a general framework encompassing the (linear) SERM problem and its convergence properties analysis; 
$(ii)$ to extend the parametric regularized estimation approaches investigated in the literature, leading to {improved} versatility for the parameter selection through the definition of wider domains; 
$(iii)$ to yield optimal solutions in order to compute the desired distributed estimates for any given topology. 
This turns particularly beneficial for a class of small-world networks, previously lacking of an effective optimal tuning method.

\medskip 
\noindent{\textbf{Paper structure} -} 
The statement of SERM problem for a given SN is provided in Section~\ref{sec:state_estimation}. %
The solving approach is investigated in Section~\ref{sec:GRDS} where the GRDS is introduced and its convergence properties are discussed.  Section~\ref{sec:advantagesofGRDS} highlights the advantages of the proposed GRDS in relation to a larger variety of topologies. 
Numerical results supporting the theoretical findings are illustrated in  Section~\ref{sec:simulations}. 
Finally, Section~\ref{sec:conclusions} concludes the work summarizing the main points.

	\medskip 
	\noindent{\textbf{Notation} -} In this work, 
	$\mathds{1}_n \in \mathbb{R}^n$ and $\mathbf{0}_{n} \in \mathbb{R}^n$ denote the $n$-dimensional (column) vectors having respectively all one and all zero entries, while $\mathbf{I}_n \in \mathbb{R}^{n\times n}$ identifies the identity matrix of dimension $n$. {The sets $\mathcal{I}$, $\mathcal{I}_0$ and $\mathcal{I}_1$ correspond to the indexes collections $\{ 1 \ldots n\}$, $\{ 0 \ldots n-1\}$, and $\{ 1 \ldots n-1\}$ respectively.}
	Then, for any square matrix $\mathbf{M}\in \mathbb{R}^{n \times n}$, {$\Lambda^{\mathbf{M}} = \{\lambda_i^\mathbf{M}, i \in \mathcal{I}_0 \} $ denotes its spectrum, i.e., the set  of  its  ordered eigenvalues, 
	and $\varsigma_{\M}$ is defined as $\varsigma_{\M} = (\lambda_1^{\M} + \lambda_{n-1}^{\M})/2$}. The notation $\mathbf{M} \succeq 0$ ($\mathbf{M} \succ 0$) implies that $\lambda_i^\mathbf{M}\geq 0$ ($\lambda_i^\mathbf{M} > 0$) for all {$i \in \mathcal{I}_0$}; if $\mathbf{M}$ is also symmetric, it is positive semidefinite (definite).
	Given $\mathbf{M}_{1}, \mathbf{M}_{2} \in \mathbb{R}^{n \times n}$, $\mathbf{M}_{1} \succeq \mathbf{M}_{2}$ stands for $\mathbf{M}_{1}-\mathbf{M}_{2} \succeq 0$. Finally, $\mathbf{M} = \diag\left(\mathbf{m}\right) \in \mathbb{R}^{n \times n}$ refers to the diagonal matrix generated by the vector $\mathbf{m} \in \mathbb{R}^n$.

	\section{The SERM Problem}
	\label{sec:state_estimation}

{In this section, the SERM problem is formally stated 
after the introduction of the graph-based model for any given SN. Moreover,  both the centralized and decentralized solution approaches derived in~\cite{fabris2019distributed,fabris2020proximal} are recalled.}

	\subsection{Graph-based SN model}
	\label{sec:graph_model}
	
	As from the literature,
	a SN composed of $n\geq 2$ {devices} can be modeled as a graph $\mathcal{G}=\left(\mathcal{V},\:\mathcal{E}\right)$ wherein each element in the \textit{nodes set} $\mathcal{V}=\left\{v_1 \ldots v_n\right\}$ corresponds to a {sensor node} in the group, and the \textit{edges set} $\mathcal{E}\subseteq \mathcal{V}\times \mathcal{V}$ describes the {nodes'} interplay in terms of unidirectional/bidirectional sensing and/or communication interactions. Hereafter, it is assumed that there exists the edge $e_{ij}=\{v_i , v_j\} \in \mathcal{E}$, i.e.,  nodes $v_i$ and $v_j$ are adjacent, if and only if the $i$-th and $j$-th {sensor nodes} can sense each other and are able to reciprocally exchange information according to some predetermined communication protocol. This implies that the graph $\mathcal{G}$ representing the given SN is \textit{undirected}. 
	In addition, w.l.o.g., it is also assumed to be \textit{connected}.

According to the graph-based SN model, the {nodes} interplay can be summarized through the \textit{adjacency matrix}~$\mathbf{A} \in \mathbb{R}^{n \times n}$ such that $[\mathbf{A}]_{ij}=1$, if the $i$-th and $j$-th {nodes} can interact; $[\mathbf{A}]_{ij}=0$, otherwise. 
Thus, for each node $v_i \in \mathcal{V}$, the set $\mathcal{N}_i =\left\{j\;|\; v_i, v_j\in\mathcal{V}, v_j\neq v_i, e_{ij} \in \Emc\right\} \subseteq \mathcal{V}$, namely its \textit{neighborhood}, represents the set of {nodes} interacting with the $i$-th one. 
The \textit{degree} $d_i = \vert \mathcal{N}_i \vert $ of node $v_i$, then, corresponds to the $i$-th element of the main diagonal of the \textit{degree matrix} $\mathbf{D} = \diag(\mathbf{A}\mathds{1}_n)\in \mathbb{R}^{n\times n}$. 
The latter is, in turns, involved in the definition of both the \emph{Laplacian matrix} $\mathbf{L}=\mathbf{D}-\mathbf{A} \in \mathbb{R}^{n\times n}$ and \emph{normalized Laplacian matrix} $\boldsymbol{\mathcal{L}}=\mathbf{D}^{-1/2}\mathbf{L}\mathbf{D}^{-1/2} \in \mathbb{R}^{n\times n}$.
	
Finally, $d_{m}={\min}_{i \in \mathcal{I}}\{d_i\}$ and $d_{M}={\max}_{i \in \mathcal{I}}\{d_i\}$ identify the \textit{minimum} and \textit{maximum graph degree}, respectively,  while  $\mathrm{vol}(\mathcal{G}) = \sum_{i \in \mathcal{I}} d_{i}$ corresponds to the graph volume.
Note that, for \textit{regular} graphs  it holds that $d_m=d_M=c >0$, and furthermore $c=n-1$ for \textit{complete} graphs.

	\subsection{Problem statement}
	\label{sec:theOptimizationProblem}

	Given a SN modeled as in Section~\ref{sec:graph_model},  each $i$-th {sensor node}, $i \in \mathcal{I}$, is assumed to be characterized by a one-dimensional scalar attribute $x_i \in \mathbb{R}$  identifying the \textit{$i$-th {node} state} and consisting in a measurable physical quantity, and by a set of \textit{relative  measurements}  $\mathcal{M}_i~=~\{\tilde{x}_{ij}~\in~\mathbb{R} \; | \; j \!\in \!\mathcal{N}_i\}$ where the recorded $\tilde{x}_{ij}$ corresponds to the difference between the neighbor state $x_j$ and the sensing {node} state $x_i$ corrupted by some noise.

	In this scenario, the SERM problem consists in the determination of the set $\{\hat{x}_i \in \mathbb{R}, \forall i \in \mathcal{I}\}$ that allows to best approximate the true system state $\{{x}_i \in \mathbb{R}, \forall i \in \mathcal{I}\}$ and to be consistent with the set of the existing relative measurements. Such an issue can be solved resorting to the popular LS paradigm as formalized in the following.
	
	\begin{SERMprob}
		\label{problem:distr_est_from_rel_meas} 
		Consider a $n$-{nodes} SN modeled by an undirected connected graph $\mathcal{G}= (\mathcal{V},\mathcal{E})$ and characterized by the set $\mathcal{M}= \bigcup_{v_i \in \mathcal{V}} \mathcal{M}_i$ of relative measurements. Introducing the vectors $\hat{\mathbf{x}} = \colvec{
			\hat{x}_{1} \; \ldots \; \hat{x}_{n}
		}^\top \in \mathbb{R}^{n}$, $\mathbf{x} = \colvec{
			x_{1} \; \ldots \; x_{n}
		}^\top \in \mathbb{R}^{n}$ that identify respectively the \textit{system state estimate} and the \textit{true system state}, the SERM problem  consists in the resolution of the following convex minimization
		\begin{align}
			\label{eq:cost}
			\hat{\x} = \underset{\x \in \Rset{n}}{\arg \min} \; h(\x),
		\end{align}
		where the (convex) cost function $h(\cdot) : \mathbb{R}^n \rightarrow \mathbb{R}$ is such that
		\begin{align} 
			\label{eq:cost_function}
			h(\x) & =  {\textstyle\frac{1}{2}} {\textstyle\sum}_{i \in \Vmc} {\textstyle\sum}_{j \in \Nmc_{i}}  (x_{i}-x_{j}+\tilde{x}_{ij})^{2}.
		\end{align}
	\end{SERMprob}

	\subsection{Centralized vs distributed solution}\label{sec:centralized_distributed_solution}

{As proven in~\cite{fabris2019distributed}, the (minimum norm) \textit{centralized solution} $\hat{\mathbf{x}}^\star_C \in \mathbb{R}^{n}$ of the minimization~\eqref{eq:cost} is given by
		\begin{equation}
			\hat{\x}^\star_C = {\textstyle\frac{1}{2}} \mathbf{L}^{\dagger} {\tilde{\mathbf{x}}}. 	\label{eq:centrSol}
		\end{equation}  
		with $\tilde{\mathbf{x}}  = \colvec{\textstyle \sum_{j \in \Nmc_{1}} (\tilde{x}_{j1}-\tilde{x}_{1j}) \; \cdots \;\ \sum_{j \in \Nmc_{n}} (\tilde{x}_{jn}-\tilde{x}_{nj})}^{\top} \in \Rset{n}$, $\tilde{\mathbf{x}} \notin \ker(\mathbf{L}) \setminus \left\lbrace \mathbf{0}_{n}\right\rbrace $, and $\mathbf{L}^\dagger\in\mathbb{R}^{n \times n}$  denoting the pseudo-inverse of the Laplacian matrix.}
Note that the computation of~\eqref{eq:centrSol} rests upon the knowledge of all the {nodes} interactions and all the relative measurements.

On the other hand, 
the first order optimality condition $[\nabla_{\x} h(\hat{\x})]_{i} = 0$ involving the gradient $\nabla_{\x} h(\x) \in \mathbb{R}^n$ of cost function~\eqref{eq:cost_function} with respect to $\x$ yields
	\begin{equation}\label{eq:distrUpdateScalar}
		\hat{x}_{i} =  d_i^{-1}\left({\textstyle\sum}_{j \in {\cal N}_i} \hat{x}_{j} + {\textstyle\frac{1}{2}}  {\textstyle\sum}_{j \in \Nmc_{i}} (\tilde{x}_{ji}-\tilde{x}_{ij}) \right).
	\end{equation}
	The estimation of $i$-th {sensor node} state in~\eqref{eq:distrUpdateScalar} depends exclusively on the state estimate of its neighbors and on its set of relative measurements, thus suggesting the adoption of a distributed paradigm in the minimization~\eqref{eq:cost}. In this direction, the following discrete system can be derived from~\eqref{eq:distrUpdateScalar} as an update rule for $\hat{\x} \in \mathbb{R}^{n}$ driven by the measurements:
	\begin{equation}\label{eq:distrUpdateVector1}
		\Sigma_{0}: \quad \hat{\x}(k+1) = \mathbf{F}_{0} \hat{\x}(k) + \mathbf{u}_{0}(\tilde{\mathbf{x}}),
	\end{equation}
		where $\mathbf{F}_{0}=\mathbf{D}^{-1} \mathbf{A} \in \mathbb{R}^{n \times n}$ and $\mathbf{u}_{0}(\tilde{\mathbf{x}}) = {\textstyle\frac{1}{2}}\mathbf{D}^{-1}\tilde{\mathbf{x}} \in \mathbb{R}^n$ correspond to the adjacency matrix and the measurements vector normalized by the degree matrix, respectively.
Note that, applying~\eqref{eq:distrUpdateVector1}, the $i$-th {sensor node} estimate $\hat{x}_i(k) \in \mathbb{R}$ at the \mbox{$k$-th} step affects its neighbors' estimate at the $(k+1)$-th step, but it is not considered for the recursive self-estimate.

As explained in~\cite{fabris2020proximal}, the convergence behavior of the scheme $\Sigma_0$ is determined by the second largest eigenvalue (in modulus) of the state matrix $\mathbf{F}_0$. Such a row-stochastic  matrix has real spectrum 
\begin{equation}
    \Lambda^{\mathbf{F}_0} \!=\! \{ \lambda^{\mathbf{F}_{0}}_i \in \left[-1,1\right] \; \vert \; 1 \!=\! \lambda^{\mathbf{F}_{0}}_0 \!>\! \lambda^{\mathbf{F}_{0}}_1 \!\geq\! \cdots \!\geq\! \lambda^{\mathbf{F}_{0}}_{n-1} \}, 
\end{equation} 
and this is related the (real) spectrum of normalized Laplacian matrix $\normlap$, since 
\begin{equation} \label{eq:eigsNLandF0}
\Lambda^{\normlap} = \{ \lambda^{\normlap}_i=1-\lambda_{i}^{\F_{0}} \; \vert \; 0=\lambda_{0}^{\normlap} <\lambda_{1}^{\normlap}  \leq \cdots \leq \lambda_{n-1}^{\normlap} \}.
\end{equation}
In detail, the system state estimation driven by~\eqref{eq:distrUpdateVector1} is guaranteed to converge toward the minimum norm solution~\eqref{eq:centrSol} if $\lambda_{n-1}^{\mathbf{F}_0}\neq-1$, namely $\lambda_{n-1}^{\normlap} \neq 2$,  and this condition is fulfilled if and only if the graph $\mathcal{G}$ representing the SN is not bipartite~\cite{chung1997spectral}. 
In the bipartite case, instead, the estimation convergence can be achieved only if it exists the possibility to act on the {nodes} interplay through $\mathbf{A}$ or, alternatively, by means of the introduction of a regularization parameter. 

	\section{A General Regularized Distributed Solution} 
	\label{sec:GRDS}

Regularization techniques can in general be exploited to improve the convergence properties also in non-bipartite scenarios. In the light of this fact, a novel GRDS for the SERM problem is presented here, as a first original contribution. 

	Based on the introduction of some suitable weights acting as regularization parameters, the proposed iterative approach entails the recursive use of the self-estimate, opportunely combined with the neighbors estimate and the set of locally available relative measurements. 
	From a graph-based perspective, the designed solution implies the insertion of weighted self-loops in the topology modeling a given SN.
	
	Formally, the \emph{real regularization parameter} $q_{i} \in \mathcal{Q} \subseteq \mathbb{R}$ is introduced, associated to the $i$-th {sensor node}, $\forall i \in \mathcal{I}$. Assuming {$\mathcal{Q} = (-1,1)$} and multiplying both sides of~\eqref{eq:distrUpdateScalar} by term $(1-q_{i})$ it follows that
	\begin{equation}\label{eq:qimultbothsides}
		\hat{x}_{i} = q_{i}\hat{x}_{i} + \frac{1-q_{i}}{d_i}\left(\sum\limits_{j \in {\cal N}_i} \hat{x}_{j} + \frac{1}{2}  \sum_{j \in \Nmc_{i}} (\tilde{x}_{ji}-\tilde{x}_{ij}) \right).
	\end{equation}
	 The resulting~\eqref{eq:qimultbothsides} still optimizes $h(\x)$ in~\eqref{eq:cost_function} and can be interpreted as a filter. 
	 This leads to the definition of the following iterative update rule 
	\begin{equation}\label{eq:sigmaQ}
		\Sigma_{\Q}: \quad \hat{\x}(k+1) = \F_{\Q}\hat{\x}(k) + \u_{\Q}\left(\tilde{\mathbf{x}}\right),
	\end{equation}
	where both the state matrix $\F_{\Q} \in \mathbb{R}^{n \times n}$ and the input~vector $\u_{\Q}\left(\tilde{\mathbf{x}}\right) \in \mathbb{R}^n$ depends on the (diagonal) \emph{regularization matrix} $\mathbf{Q}=\diag{\left(\mathbf{q}\right)} \in \Rset{n\times n}$ with $\mathbf{q} = \colvec{ q_1 \; \cdots \; q_n}^\top \!\in \!\mathbb{R}^n$. Specifically, it holds that
	\begin{align}\label{eq:FQ}
		&\F_{\Q} = \Q + (\I_{n}-\Q)\F_{0}, \\ \label{eq:uQ}
		& \u_{\Q}\left(\tilde{\mathbf{x}}\right) = (\I_{n}-\Q)\u_{0}(\tilde{\mathbf{x}}).
	\end{align}
 Note that the regularized scheme $\Sigma_{\Q}$ generalizes $\Sigma_{0}$: the update rules~\eqref{eq:distrUpdateScalar} and~\eqref{eq:qimultbothsides} coincide when $q_i = 0$, $\forall i \in \mathcal{I}$. 
	
			\subsection{Convergence properties}
		
	The convergence of the  scheme $\Sigma_{\Q}$ is affected by the spectral properties of the state matrix~\eqref{eq:FQ}, which, in turns, depends on the regularization matrix $\Q$. In particular, the distributed update rule~\eqref{eq:sigmaQ} is not ensured to be convergent for any choice of $\{q_i \in \mathcal{Q}, i \in \mathcal{I}\}$ as clarified in the following.

	First, observe that, for any selection of $\Q$, the spectrum of the matrix $\mathbf{F}_\Q$ includes the eigenvalue $\lambda_{0}^{\F_{\Q}}=1$. In addition, the following properties also hold. 
			\begin{lem}\label{lem:eigsFQareReal}
		For any $\Q \in \mathbb{R}^{n \times n}$, all the eigenvalues of the corresponding matrix $\F_{\Q}$ are real, i.e., $\lambda_{i}^{\F_{\Q}} \in \mathbb{R}$ $ \forall i\in \Imc_{0}$.
	\end{lem}
	\begin{proof}
	The spectrum of a matrix $\N \in \mathbb{R}^{n \times n}$ is real if it exists a positive definite diagonal matrix $\M \in \mathbb{R}^{n \times n}$ 
		such that $\M \N = \N^{\top} \M$. 
		Therefore, since $(\I-\Q)^{-1}\D \succ 0$ for any choice of $\Q$,  the thesis is proven selecting $\M = (\I-\Q)^{-1}\D $ and $\N = \F_\Q$.
	\end{proof}
	\begin{lem}\label{lem:boundsforFQ}
		For any $\Q \in \mathbb{R}^{n \times n}$, the eigenvalues of the corresponding matrix $\F_\Q$ are bounded. In particular, it holds 		\begin{equation}\label{eq:eigenboundsq}
			{\lambda}_{i}^{\F_{\underline{q}}} \leq \lambda_{i}^{\F_{\Q}} \leq {\lambda}_{i}^{\F_{\overline{q}}},  \quad \forall i \in \mathcal{I}_i,
					\end{equation}
					with $\F_{\underline{q}}=\underline{q}\I_{n}+(1-\underline{q}) \F_{0}$ and $\F_{\overline{q}}=\overline{q}\I_{n}+(1-\overline{q}) \F_{0}$, given 
								 $\underline{q}= \min_{i \in \mathcal{I}} q_{i}$ and $\overline{q}= \max_{i \in \mathcal{I}} q_{i}$.
	\end{lem}
	\begin{proof}
		Considering the matrix $\F_{\Q}$ as a function of $\Q$ and  exploiting the relation~\eqref{eq:eigsNLandF0},  the Gershgorin circle theorem 
		 ensures that the condition $\overline{q}\eye{n} \succeq \Q \succeq \underline{q} \eye{n}$  implies
		\begin{equation}\label{eq:boundsofQthroughF0proof}
			\begin{cases}
				(\Q-\underline{q}\eye{n}) (\eye{n}-\F_{0}) \succeq 0, \\
				(\overline{q}\eye{n}-\Q) (\eye{n}-\F_{0}) \succeq 0.
			\end{cases}
		\end{equation}
		Then, adding $\F_{0}$ to both sides of inequalities~\eqref{eq:boundsofQthroughF0proof} and applying Lemma \ref{lem:eigsFQareReal}, it follows that
		\begin{equation}
		\hspace{-6pt}
			\begin{cases}
				\F_{\Q}=\Q+(\I-\Q)\F_{0} \succeq \underline{q}\eye{n}+(1-\underline{q})\F_{0} = \F_{\underline{q}}, \\
				\F_{\overline{q}} = \overline{q}\eye{n}+(1-\overline{q})\F_{0} \succeq \Q+(\I-\Q)\F_{0}=\F_{\Q}.
			\end{cases} 
		\end{equation}
		As a consequence, for any $\Q$, it holds that $\F_{\overline{q}}  \succeq \F_\Q \succeq \F_{\underline{q}}$ and, thus, the spectral relation~\eqref{eq:eigenboundsq} is guaranteed. 
	\end{proof}
	
	Then, next theorem shows that, adopting the  scheme $\Sigma_{\Q}$, the system state estimate  asymptotically converges towards the centralized solution~\eqref{eq:centrSol} only for a suitable selection of the regularization parameters set that ensures some spectral properties  of the state matrix  $\F_{\Q}$.

\begin{table*}[t!]
		\begin{center}
		\resizebox{\textwidth}{!}{
			\renewcommand{\arraystretch}{1.5}
			\begin{tabular}{c||l|l|l||l|l|}
				\cline{2-6}
 &	\multicolumn{3}{c||}{state-of-the-art} & 	\multicolumn{2}{c|}{new contribution} \\
				\cline{2-6}
				& State Matrix (SM) & {Input Vector} & Parameter Domain & Parameter Ext Domain & Opt Parameter Selection \\  
								\hline \hline
	\multicolumn{1}{|c||}{{$\Sigma_\eta$}} &  $\F_\eta = \eta \I_n +(1-\eta)\F_{0}$ &  $\u_{\eta} = (1-\eta)\u_{0}(\tilde{\mathbf{x}})$ & $\mathcal{Q}_\eta=\left(0,1\right)$  & {$\widecheck{\mathcal{Q}}_\eta=\left(1-\frac{2}{\lambda_{n-1}^{\normlap}}, 1\right)$} & { $\eta^{\ast} = 1-\varsigma_{\normlap}^{-1}$}  \\ 
								\hline
				\multicolumn{1}{|c||}{{$\Sigma_\rho$}} &  $\F_\rho = \left( \D + {\textstyle\frac{\rho}{2}} \I_{n}  \right)^{-1} \left(  \A + {\textstyle\frac{\rho}{2}} \I_{n}    \right)$  &  {$\u_{\rho} =  \left( \D + {\textstyle\frac{\rho}{2}} \I_{n}  \right)^{-1} \D \u_{0}(\tilde{\mathbf{x}})$}& $\mathcal{Q}_\rho=\left(0,+\infty\right)$ &  $\widecheck{\mathcal{Q}}_\rho=\left(d_{m}\left(\lambda_{n-1}^{\normlap}-2\right),+\infty\right)$ & { $\rho^{\ast} \in 2(\varsigma_{\normlap}-1)\cdot\left[s_{m},s_{M}\right]$}\\
				\hline
				\multicolumn{1}{|c||}{{$\Sigma_\epsilon$}} &  $\F_{\epsilon} = \eye{\nonodes}-\epsilon \lap$ & {$\u_{\epsilon} =  \epsilon \D  \u_{0}(\tilde{\mathbf{x}})$} & $\mathcal{Q}_\epsilon=\left(0,\frac{1}{d_M}\right)$ & {$\overline{\mathcal{Q}}_\epsilon = \left(0,\frac{2}{\lambda_{n-1}^{\L}}\right)$}  & { $\epsilon^{\ast} = \varsigma_{\lap}^{-1}$} \\
				\hline
			\end{tabular}
		}
	\end{center}
	\caption{main features of the $\Sigma_{\eta}$, $\Sigma_{\rho}$, and $\Sigma_{\epsilon}$ proposed in~\cite{fabris2019distributed,fabris2020proximal,olfati2007consensus} and interpreted as GRDS realizations - the optimal selection of $\rho$ rests on a binary-search in the reported interval wherein $s_{m} = d_{m}$ and $s_{M} = d_{M}$  if $\varsigma_{\normlap} \geq 1$ and $s_{m} = d_{M}$ and 	$s_{M} = d_{m}$ otherwise.
	}
	\label{tab:schemes}
\end{table*}

\begin{thm} \label{thm:convergencetocentr}
		{If matrix $\Q$ is selected so that for the corresponding matrix $\F_{\Q}$ it results 
		\begin{equation}
		    \vert \lambda_{i}^{\F_{\Q}}\vert < 1, \quad \forall i \in \mathcal{I}_1, \label{eq:convergence_condition}
		\end{equation}
		then the system state estimation driven by~\eqref{eq:sigmaQ}  asymptotically converges to the centralized solution~\eqref{eq:centrSol} in terms of relative differences.
		Specifically, under the given condition~\eqref{eq:convergence_condition}, it holds that $\hat{x}_{\Q,i}^\star-\hat{x}_{\Q,j}^\star = \hat{x}_{C,i}^{\star}-\hat{x}_{C,j}^{\star}$, $\forall i,j \in \mathcal{I}$ such that $e_{ij} \in \Emc$, with $\hat{\mathbf{x}}_C^\star = \colvec{\hat{x}_{C,1}^\star\; \cdots\; \hat{x}_{C,n}^\star}^\top$ defined as in~\eqref{eq:centrSol} and $\hat{\mathbf{x}}^\star_\Q\in \mathbb{R}^n$ such that $\hat{\mathbf{x}}^\star_\Q = \colvec{\hat{x}_{\Q,1}^\star\; \ldots \;\hat{x}_{\Q,n}^\star}^\top= \lim_{k \rightarrow +\infty} \hat{\x}(k)$ with $\hat{\x}(k)$ evolving according to~\eqref{eq:sigmaQ}.}
	\end{thm}
	\begin{proof}
	Accounting for the scheme $\Sigma_\Q$, one can verify that, given the initial condition $\hat{\x}(0) = \hat{\mathbf{x}}_{0}$, it holds
	\begin{equation}
		     \hat{\x}(k) = \mathbf{F}_{\Q}^{k} \hat{\mathbf{x}}_{0} + \left( \textstyle\sum_{l=0}^{k-1} \F_{\Q}^{l} \right) \mathbf{u}_{\Q}(\tilde{\x}) \label{eq:limit}
		    \end{equation}
and, at the equilibrium, one has $\hat{\mathbf{x}}_{eq} = \mathbf{F}_{\Q} \hat{\mathbf{x}}_{eq} + \mathbf{u}_{\Q}\left(\tilde{\mathbf{x}}\right)$. Taking into account~\eqref{eq:FQ}-\eqref{eq:uQ}, this implies
\begin{equation}
2\L \hat{\mathbf{x}}_{eq}= \tilde{\mathbf{x}}.  \label{eq:eq_condition}
\end{equation}
In particular, condition~\eqref{eq:eq_condition}  ensures that  a generic equilibrium solution for the scheme $\Sigma_\Q$ can be expressed as $\hat{\mathbf{x}}_{eq} = \hat{\x}^{\star}_C + \alpha \ones$,	where {$\hat{\x}^{\star}_C$ is the minimum norm centralized solution~\eqref{eq:centrSol}, while $ \alpha \in \Rset{}$ and $\{\ones\} = \ker(\mathbf{L})$}. \\
		Now,  by the Perron-Frobenius theory, condition~\eqref{eq:convergence_condition} is sufficient to ensure that $\lim_{k \rightarrow \infty}  \mathbf{F}_{\Q}^{k} =  \ones \mathbf{v}_{\Q}^{\top}$, 
	where $\mathbf{v}_{\Q} \in \Rset{n}$  denotes the (left) eigenvector  associated to the eigenvalue $\lambda^{\F_{\Q}}_{0} = 1$. 
	Then, under the equilibrium condition~\eqref{eq:eq_condition}, the limit $\lim_{k \rightarrow +\infty} \hat{\x}(k)$ converges towards a unique $\hat{\mathbf{x}}^\star_\Q $ since it follows from~\eqref{eq:limit}  that 
			\begin{align}
	    \lim_{k \rightarrow \infty} \hat{\x}(k) 
	    &=\left( \ones \mathbf{v}_{\Q}^{\top} \right) \hat{\mathbf{x}}_{0} \!+\!  \lim_{k \rightarrow \infty}  \left( (\eye{n}-\mathbf{F}_{\Q}^{k}) \hat{\mathbf{x}}_{eq} \right) \\
	    	    &=\left( \ones \mathbf{v}_{\Q}^{\top} \right) \hat{\mathbf{x}}_{0} \!+\! \left( \eye{n}-\ones \mathbf{v}_{\Q}^{\top} \right) \hat{\mathbf{x}}_{eq} 
	\end{align}
In particular, exploiting the fact that $\hat{\mathbf{x}}_{eq} = \hat{\x}^{\star}_C + \alpha \ones$, 
it yields	
		\begin{align}
			\hat{\mathbf{x}}_\Q^\star 
			&= \hat{\x}^{\star}_C + \ones \left (\mathbf{v}_{\Q}^{\top} \hat{\mathbf{x}}_{0} - \mathbf{v}_{\Q}^{\top} \hat{\x}^{\star}_C - \alpha \mathbf{v}_{\Q}^{\top}\ones + \alpha\right ) \\
			&= \hat{\x}^{\star}_C + \beta \ones,
		\end{align}
		where  $\beta \in \mathbb{R}$: the system state estimate provided by the scheme $\Sigma_{\Q}$  thus converges to the solution~\eqref{eq:centrSol}  in terms of relative differences.
		\end{proof}

	\subsection{Regularization parameter selection}
	
In the rest of the section, the attention is focused on the conditions on the set of parameters $\{q_i, \forall i \in \mathcal{I}\}$ guaranteeing the  validity of Theorem~\ref{thm:convergencetocentr}.
In this direction, in the light of Lemma~\ref{lem:boundsforFQ}, one can realize that some \textit{limit cases} exist for the fulfillment of condition~\eqref{eq:convergence_condition}. These correspond to the selection of the regularization matrix as a scalar matrix, namely $\Q = q \I_{n}$ (with $q<1$). Accounting for this fact, the next theorem 
	provides some (conservative) requirements on the regularization parameters selection sufficient to ensure the convergence of scheme $\Sigma_\Q$.

	\begin{thm} \label{thm:sufficientcondsimplstab}
		Condition~\eqref{eq:convergence_condition} in Theorem~\ref{thm:convergencetocentr} is satisfied for any diagonal matrix $\Q$ defined by selecting the $q_i,\forall i\in~\mathcal{I}$ 
		\begin{equation}\label{eq:constraintqiU}
			q_{i} \in \widecheck{\mathcal{Q}}=\left(\mu, 1 \right), \;\; \mu = 1-\frac{2}{\lambda_{n-1}^{\normlap}} \in \left[ -1+\frac{2}{n} , 0\right], 
		\end{equation}
		where $\lambda_{n-1}^{\normlap}$ is the largest eigenvalue of the normalized Laplacian, and $\mu = -1+2/n$ and $\mu = 0$ hold for the complete and bipartite graphs, respectively.
	\end{thm}
	\begin{proof}
	Based on~\eqref{eq:eigenboundsq}, condition~\eqref{eq:convergence_condition} is fulfilled only if $\lambda_i^{\F_{\underline{q}}} > -1$ and $\lambda_i^{\F_{\overline{q}}} <1$, $\forall i\in \Imc_{1}$. In the limit case $\Q = \overline{q} \I_n$, with $\overline{q} <1$,  the given requirement reduces to the condition $\lambda_i^{\F_{\underline{q}}}=\underline{q}+(1-\underline{q})\lambda_{i}^{\F_{0}}  > -1$, $\forall i\in \Imc_{1}$ that, given~\eqref{eq:eigsNLandF0}, implies 
		\begin{equation}
			\underline{q} > \mu=(\lambda^{\F_{0}}_{n-1}+1)/(\lambda^{\F_{0}}_{n-1}-1) = 1-2/\lambda_{n-1}^{\normlap}.
		\end{equation}
		To conclude, the parameter $\mu$ is an increasing function of the eigenvalue $\lambda_{n-1}^{\normlap} \in \left[n/(n-1), 2\right]$, which assumes extremal values in case of complete and bipartite graphs.
	\end{proof}

\section{Performance Assessment of the GRDS}
\label{sec:advantagesofGRDS}

Due to its generalized nature, the proposed GRDS scheme permits to extend the structure of the three distributed regularized schemes $\Sigma_{\eta}, \Sigma_{\rho}, \Sigma_{\epsilon}$, introduced in~\cite{fabris2019distributed,fabris2020proximal,olfati2007consensus}.
Interestingly, this extension pursues a twofold aim, namely $i)$ to serve as a performance assessment of the GRDS approach with respect to the existing literature, and $ii)$ to improve the convergence performance of the original schemes $\Sigma_{\eta}, \Sigma_{\rho}, \Sigma_{\epsilon}$ if revisited within the GRDS framework. 

Indeed, as summarized in Table~\ref{tab:schemes}, it is straightforward that $\Sigma_{\eta}, \Sigma_{\rho}, \Sigma_{\epsilon}$ can be interpreted as particular realizations of the novel  $\Sigma_\Q$, where the entries of $\Q$ are defined by selecting
\begin{itemize}
    \item[$\diamond$] $q_{i}=\eta$, $\forall i \in \mathcal{I}$, for $\Sigma_{\eta}$, 
    \item[$\diamond$] $q_{i} = (\frac{\rho}{2})(d_{i}+\frac{\rho}{2})^{-1}$,  $\forall i \in \mathcal{I}$, for $\Sigma_{\rho}$, 
    \item[$\diamond$] $q_{i}=1-\epsilon d_{i}$, $\forall i \in \mathcal{I}$, for $\Sigma_{\epsilon}$.
\end{itemize}
It has been already shown in~\cite{fabris2019distributed,fabris2020proximal,olfati2007consensus} that the adoption of any scheme $\Sigma_{\eta}, \Sigma_{\rho}, \Sigma_{\epsilon}$ ensures the convergence of the system state estimate towards the centralized solution~\eqref{eq:centrSol} when the corresponding regularization parameter  is selected in the domain $\mathcal{Q}_\bullet$, with $\bullet$ standing for $\eta, \rho$ or $\epsilon$ (third column of Table~\ref{tab:schemes}).
In addition, since the estimation convergence speed is linked to the second largest (in modulus) eigenvalue of the related state matrix, the choice of such parameters can be optimized based on the minimization of the Convergence Rate Index (CRI), defined as
\begin{equation}\label{eq:minimizationoftheratecriterion}
r_{\bullet}=\underset{i \in  \mathcal{I}_1}{\max} \vert\lambda_{i}^{\mathbf{F}_{\bullet}} \vert, \quad \rate{\bullet} \in[0,1].
\end{equation}
In particular, for the schemes $\Sigma_{\eta}, \Sigma_{\rho}$, this selection is valid for graphs with $\varsigma_{\normlap}\ge 1$, while otherwise a trivial solution is enforced with $\eta=0$ and $\rho=0$, respectively~\cite{fabris2019distributed,fabris2020proximal}.

Now, when accounting for the GRDS interpretation, the convergence of the existing solutions is guaranteed on an extended regularization parameter domain $\widecheck{\mathcal{Q}}_\bullet \supseteq  \mathcal{Q}_\bullet$, reported in the fourth column 
of Table~\ref{tab:schemes}: $\widecheck{\mathcal{Q}}_\eta$ and $\widecheck{\mathcal{Q}}_\rho$ are computed according to~\eqref{eq:constraintqiU}, while $\overline{{\mathcal{Q}}}_\epsilon \supseteq \widecheck{\mathcal{Q}}_\epsilon$ is derived from the fact that Theorem~\ref{thm:sufficientcondsimplstab} provides only sufficient conditions.
Remarkably, this extension is beneficial for the improvement of the convergence rate of the estimation algorithms. 
In actual facts, by taking into account the introduced extended domains, it is possible to prove that the optimal parameter selection results as in the last column of Table~\ref{tab:schemes} and it is valid for {\emph{any}} undirected and connected graph modeling a given SN. 
The latter fact also constitutes an original contribution of this work with respect to the existing literature, given that it proposes a GRDS scheme $\Sigma_\Q$ for graph topologies with any $\varsigma_{\normlap}$ and extends the applicability of $\Sigma_{\eta}, \Sigma_{\rho}$ for $\varsigma_{\normlap} < 1$.

\subsection{Discussion over graph topologies with $\varsigma_{\normlap} < 1$}
\label{sec:sigma}

The property $\varsigma_{\normlap} < 1$ is not characterizing a unique class of graphs, however it suggests some interesting considerations on a variety of topologies.  
Indeed, 
for graphs whose eigenvalues can be computed in closed form, the above condition can be verified straightforwardly. For example, it is possible to assess that for any circulant graph $C_{n}(1,2)$ wherein $\mathcal{V} = \{ v_1 \ldots v_n\}$ with $n>10$ and  $\mathcal{E} =\{ (v_i, v_{(i\pm1)\text{mod} n}),..., (v_i,v_{(i\pm 2) \text{mod}n}), \forall i \in \mathcal{I}\}$, it holds that $\varsigma_{\normlap} < 1$. 
More in general, $\varsigma_{\normlap} = (\lambda_{1}^{\normlap}+\lambda_{n-1}^{\normlap})/2$ can be upper-bounded by means of the bounds on the normalized Laplacian matrix eigenvalues.
Specifically, 
it holds that
\begin{itemize}
    \item[$\diamond$] $\lambda_{1}^{\normlap} \leq 2C_{\Gmc}$, where $C_{\Gmc}=  \min_{\mathcal{G}_0 } \frac{ |\bar{\mathcal{E}}(\mathcal{G}_0,\mathcal{G}_1)|}{\min\{\mathrm{vol}(\mathcal{G}_0),\mathrm{vol}(\mathcal{G}_1)\} }$ is the \emph{Cheeger constant} of the graph $\Gmc= (\mathcal{V},\mathcal{E})$. $C_{\Gmc}$ is defined accounting for the subgraph $\mathcal{G}_0= (\mathcal{V}_0,\mathcal{E}_0)$ of $\Gmc$ and its complement $\mathcal{G}_1= (\mathcal{V}_1,\mathcal{E}_1), \mathcal{V}_1 = \mathcal{V} \backslash \mathcal{V}_0$,  through the cardinality of their cut set $\bar{\mathcal{E}}(\mathcal{G}_0,\mathcal{G}_1) =\{(v_i,v_j) \in \mathcal{E} \; \vert \; v_i \in \mathcal{V}_0, v_j \in \mathcal{V}_1\} \subseteq \mathcal{E}$ and their volume~\cite{chung1997spectral};
    \item[$\diamond$] $\lambda_{n-1}^{\normlap} \leq 2(1- H_{\Gmc})$, where $H_{\Gmc} = \min_{\mathcal{H}} \frac{|\Nmc_{i} \cap \Nmc_{j}|}{2\max\{d_{i},d_{j}\}}$  with $\mathcal{H} = \{ i, j \in \mathcal{I} \; \vert \; i<j,  (v_i,v_j) \in \Emc\}$~\cite{rojo2013new}.
\end{itemize}
It follows by simple substitution that the property $\varsigma_{\normlap} < 1$ is ensured for all the topologies such that $C_{\Gmc} < H_{\Gmc}$, which may be useful when the eigenvalues are not easy to be computed\footnote{This condition is less strict than that over the eigenvalues: in the case above of the circulant matrix $C_{n}(1,2)$ it yields $C_{\Gmc}<H_{\Gmc}$ if $n>27$.}.
Intuitively, $C_{\Gmc}$ occurs to be small for graphs with high $n$ that can be partitioned into complementary subgraphs $\mathcal{G}_0$ and $\mathcal{G}_1$, approximating complete topologies (hence with large volumes) associated to cut sets $\bar{\mathcal{E}}(\mathcal{G}_0, \mathcal{G}_1)$ having small cardinality. %
Note that these features  generally characterize the \emph{small-world} patterns. 
Conversely, being $H_{\Gmc}\in \left[0,\frac{n-2}{2(n-1)} \right]$, high values of $H_{\Gmc}$ are obtained for high network cardinality $n$, thus decreasing the bound for $\lambda_{n-1}^{\normlap}$\footnote{This is in line with the bound given by $\lambda_{1}^{\normlap} \le \frac{n}{n-1} \le \lambda_{n-1}^{\normlap}$, where the lower bound for $\lambda_{n-1}^{\normlap}$ is minimized, tending to $1$, when $n$ is high.}.

\section{Numerical Results}\label{sec:simulations}
	
Some simulation results are presented in this section with the intent of highlighting the strengths of the proposed system state estimation GRDS approach with respect to the state-of-art solutions described in~\cite{fabris2019distributed,fabris2020proximal,olfati2007consensus}, especially in case of {SNs} modeled by graphs  characterized by $\varsigma_{\normlap}<1$.
	
For this purpose, the attention is focused on four case studies, where the considered {SNs} are associated to the topologies 
in Figure~\ref{fig:sim_convrates}. In correspondence to these scenarios, the same figure reports the CRI~\eqref{eq:minimizationoftheratecriterion}, computed by adopting
\begin{itemize}
    \item[$\diamond$] the scheme $\Sigma_0$, 
    \item[$\diamond$] the schemes $\Sigma_\eta$, $\Sigma_\rho$ and $\Sigma_\epsilon$, optimizing the regularization parameters selection both on $\mathcal{Q}_\bullet$ and on $\widecheck{\mathcal{Q}}_\bullet$, 
        \item [$\diamond$] the scheme $\Sigma_\Q$ with $\Q$ derived either from a random selection of $q_i \in \widecheck{\mathcal{Q}}$,  $i\! \in \! \mathcal{I}$, or from an iterative optimization over several trials: in the latter, the parameters $q_i$, $i\! \in \! \mathcal{I}$, are initialized accounting for the best (in terms of convergence speed) of the previous strategies and updated based on a greedy heuristic after random perturbations.
	\end{itemize}
Two preliminary observations are required: for the scheme $\Sigma_\epsilon$, it holds that $\epsilon^{\ast} = \varsigma_{\L}^{-1}$ independently on the parameter domain,
hence in Figure~\ref{fig:sim_convrates} a single value ${r}_{\epsilon^\ast}$ is reported;  $r_{\Q^\ast}$, instead, refers to the CRI computed in correspondence to the results of the optimization of the parameters $q_i, i \in \mathcal{I}$.
	
\begin{figure*}[t!]
	\centering
	\includegraphics[width = 0.84\textwidth]{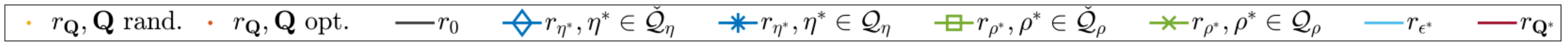}\\		\subfigure[Small-world graph -- $n = 22$, $\varsigma_{\normlap} < 1$]{
	\includegraphics[trim={0cm -5.5cm 0cm 0cm},clip,width=0.33\columnwidth]{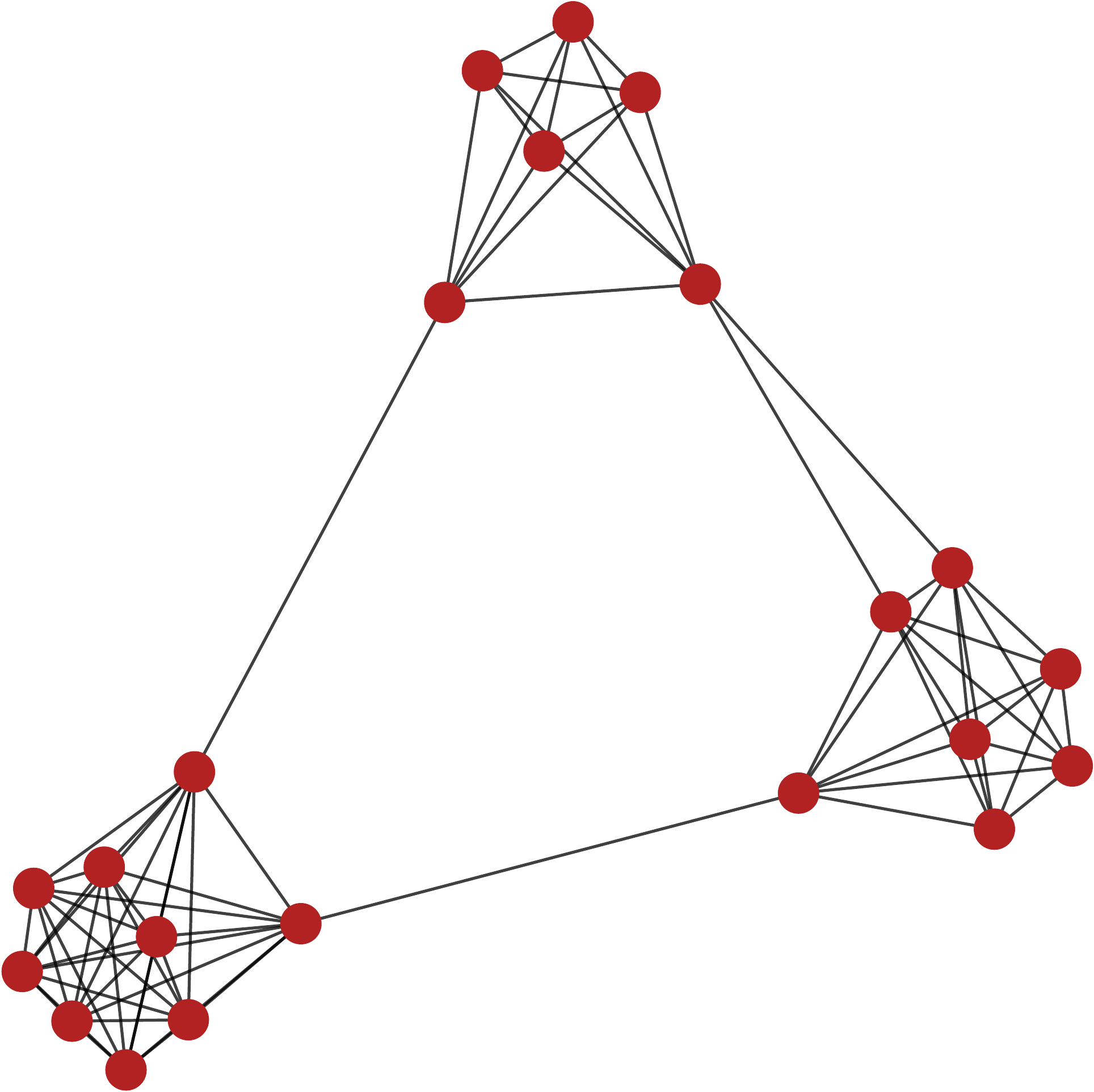} \label{fig:sim_topologiesa} \hspace{0.15cm}
	\includegraphics[trim={0 0cm 0 0cm},clip,width=0.63\columnwidth]{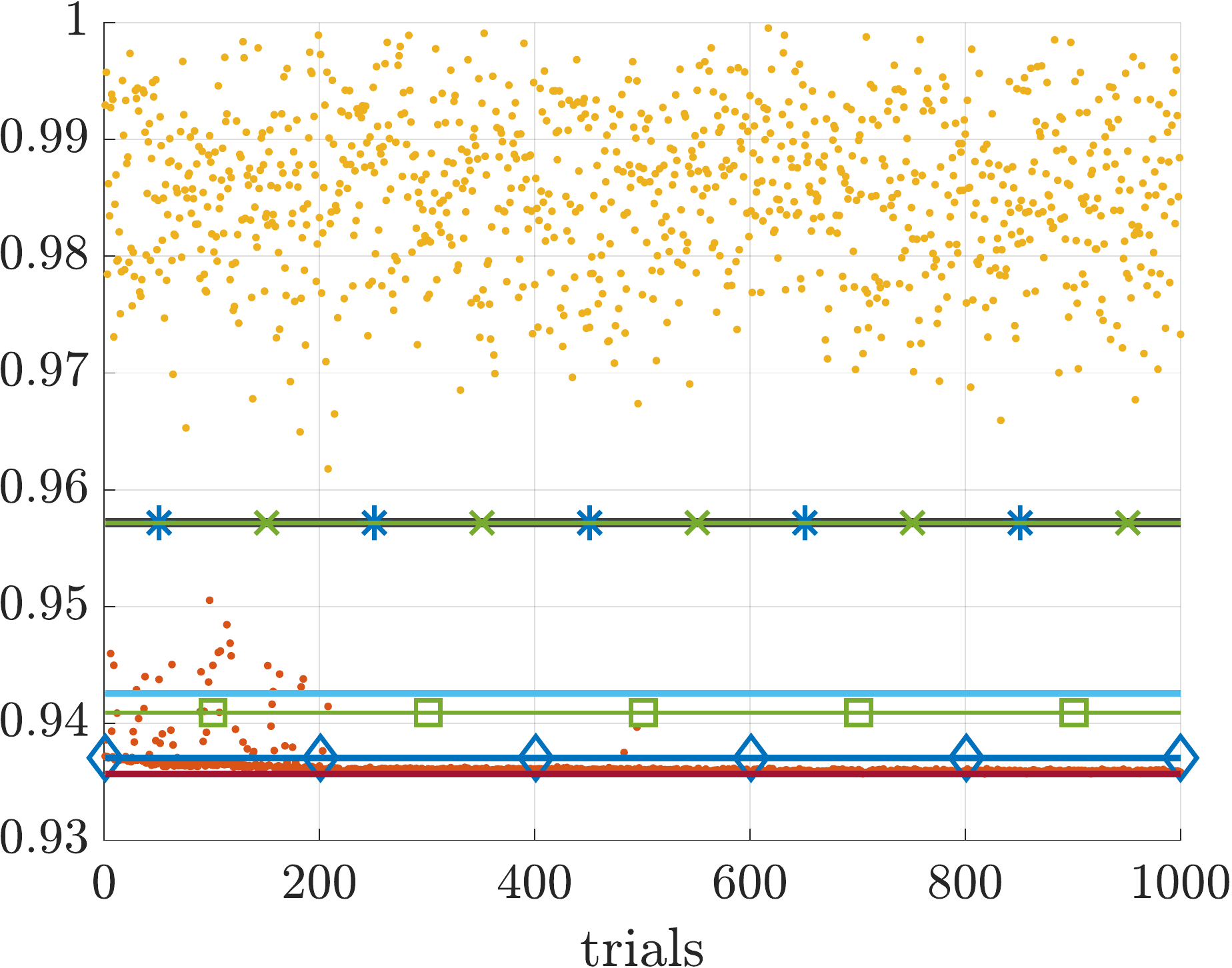}\label{fig:sim_convratesa}} \hfill 
	\subfigure[Circulant graph -- $n = 36$, $\varsigma_{\normlap} < 1$]{
	\includegraphics[trim={0cm -7.5cm 0cm 0cm},clip,width=0.33\columnwidth]{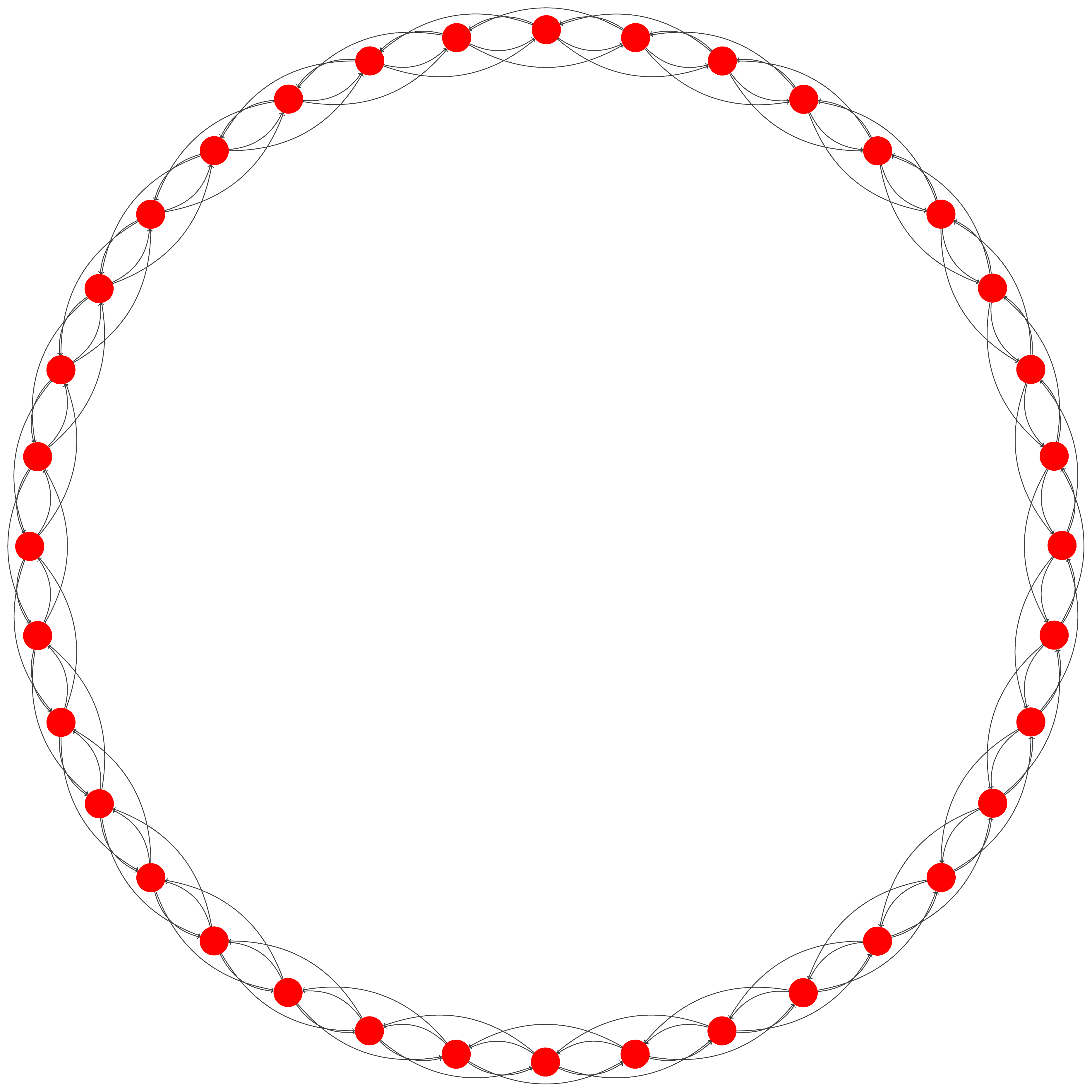}\label{fig:sim_topologiesb} \hspace{0.15cm}
	\includegraphics[trim={0 0cm 0 0cm},clip,width=0.63\columnwidth]{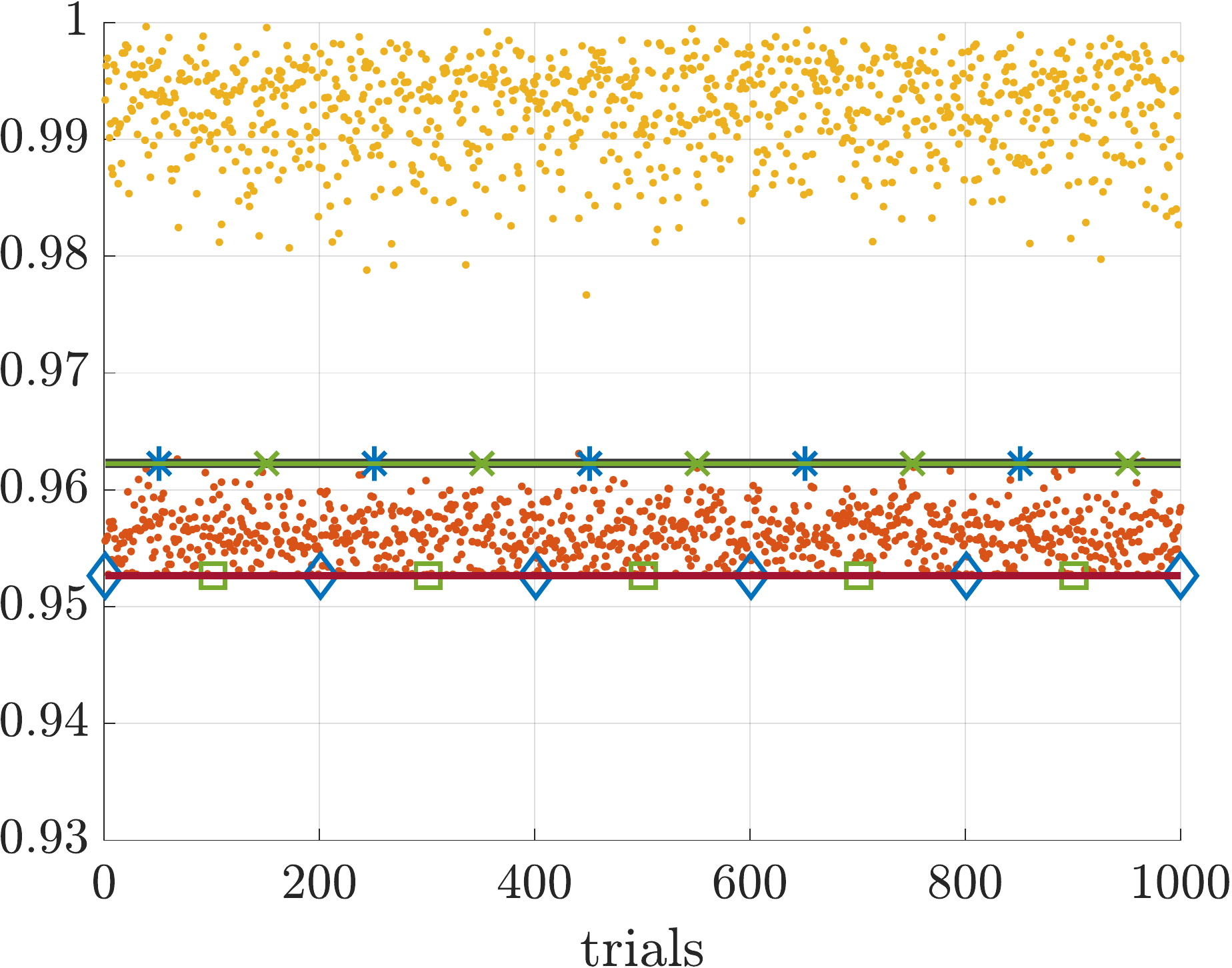}\label{fig:sim_convratesb}}
	\\
	\subfigure[Friendship graph -- $n = 19$, $\varsigma_{\normlap} = 1$]{
	\includegraphics[trim={0cm -5.5cm 0cm 0cm},clip,width=0.33\columnwidth]{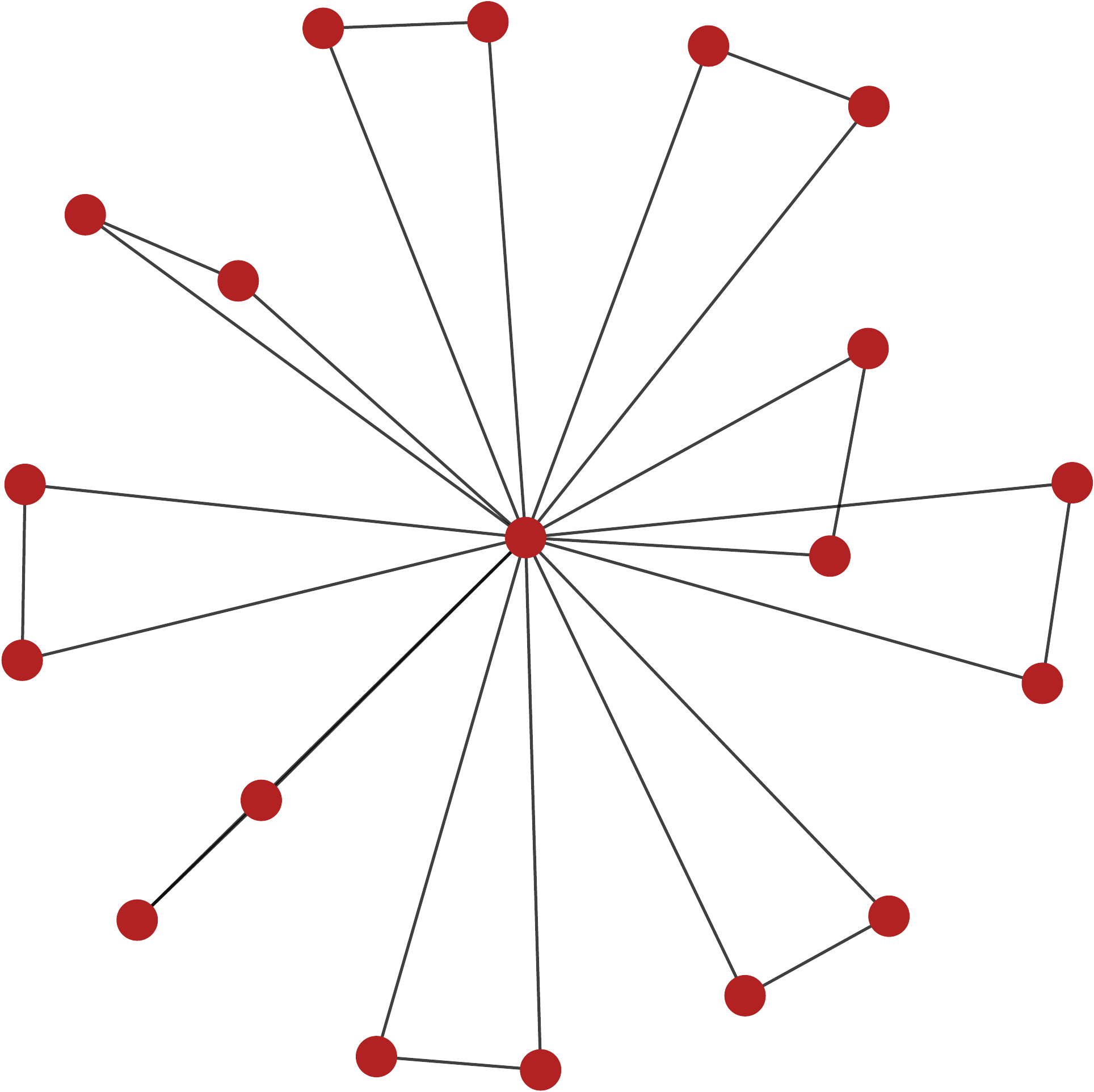}\label{fig:sim_topologiesc} \hspace{0.1cm}
	\includegraphics[trim={0 0cm 0 0cm},clip,width=0.63\columnwidth]{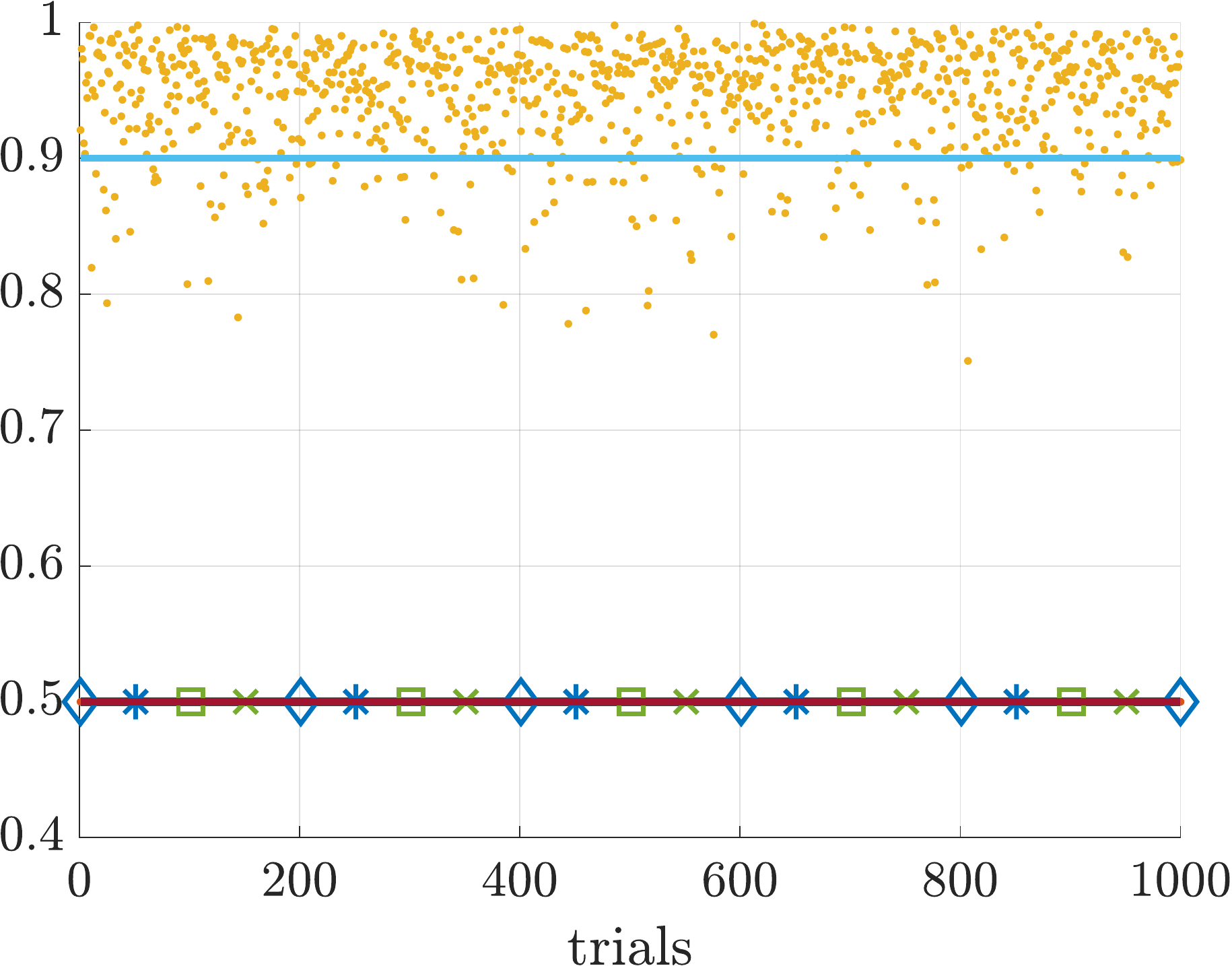}\label{fig:sim_convratesc}} \hfill
	\subfigure[Ramanujan graph -- $n = 16$, $\varsigma_{\normlap} > 1$]{
	\includegraphics[trim={0cm -3.5cm 0cm 0cm},clip,width=0.34\columnwidth]{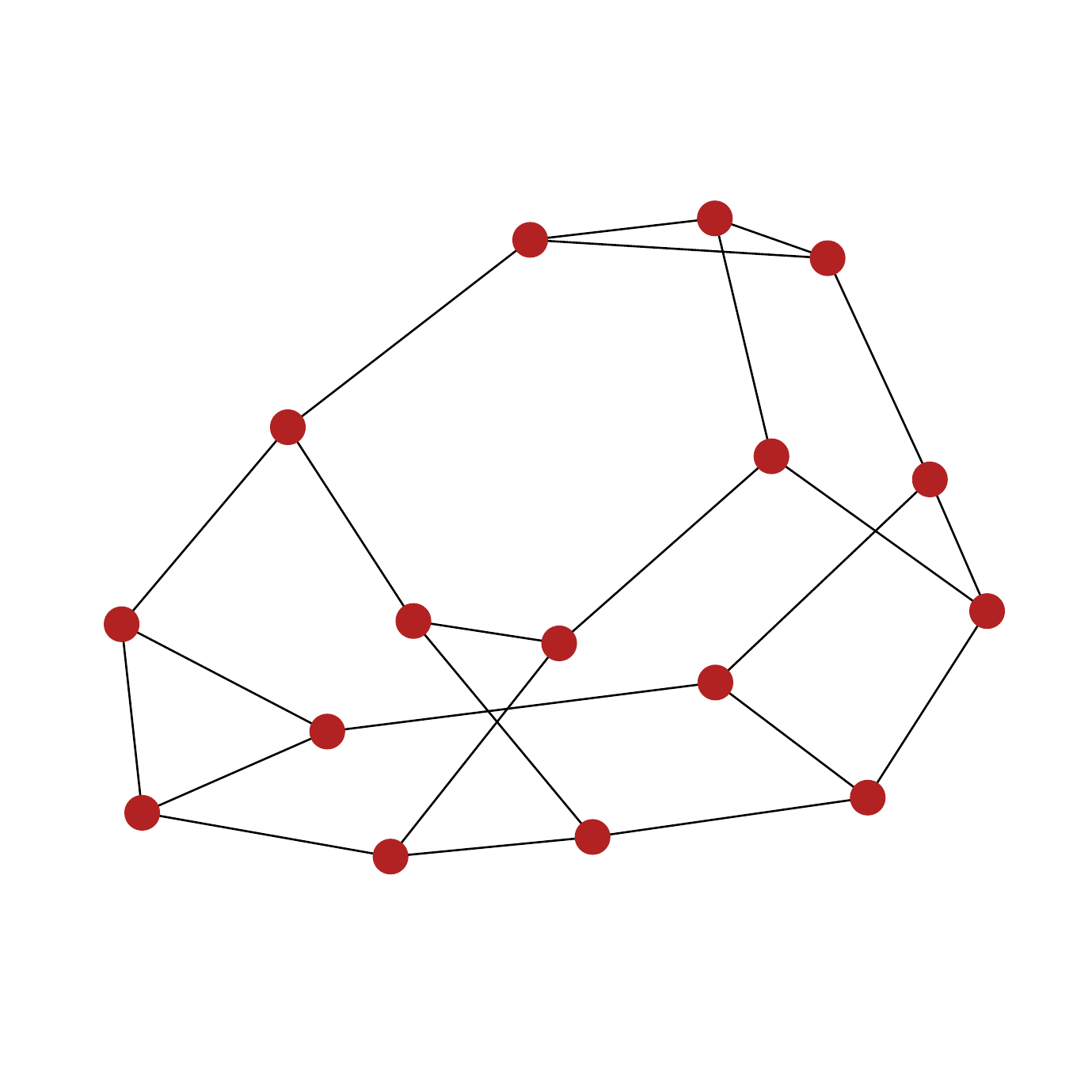}\label{fig:sim_topologiesd} \hspace{0.15cm}
	\includegraphics[trim={0cm 0cm 0cm 0cm},clip,width=0.63\columnwidth]{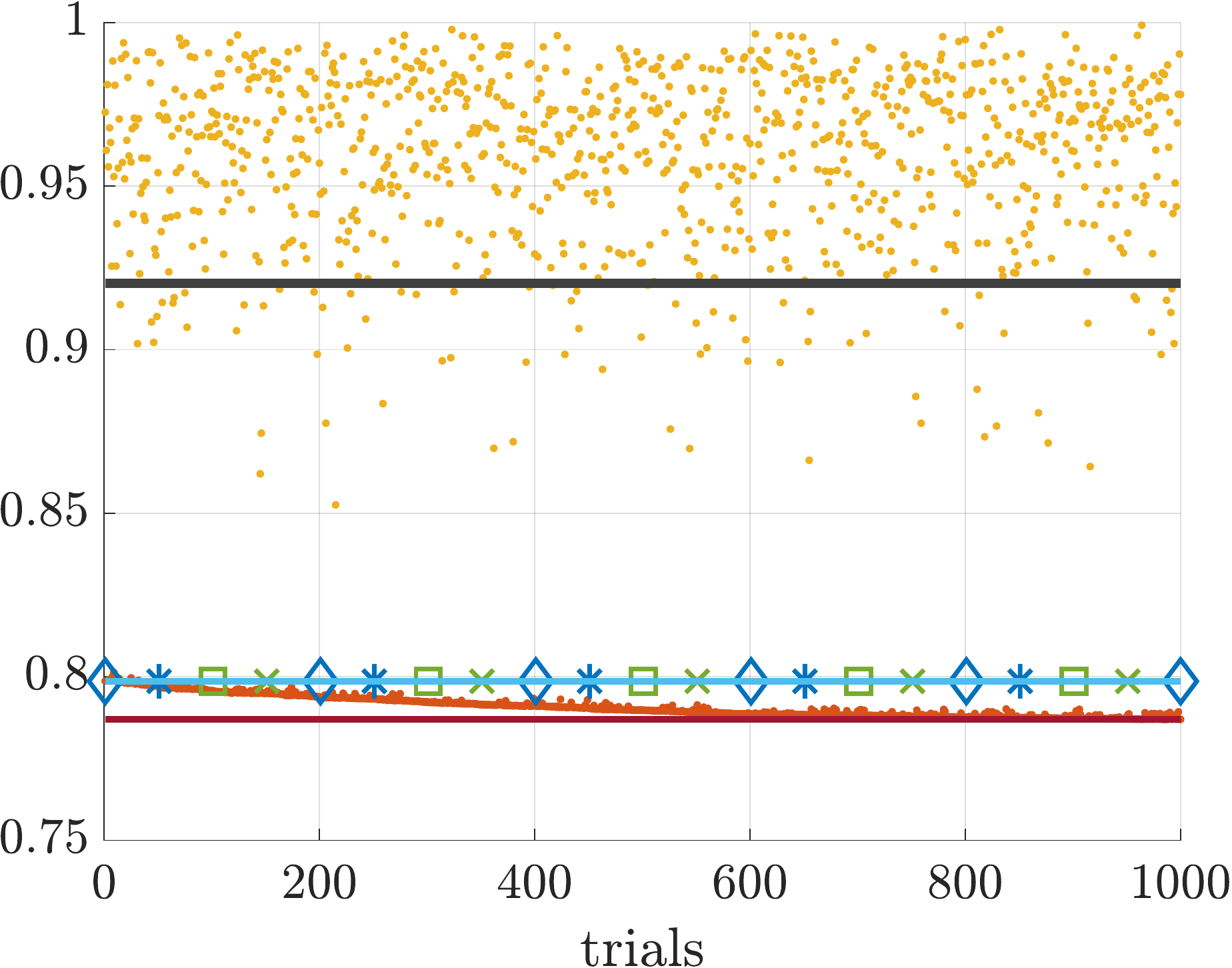}\label{fig:sim_convratesd}}
	\caption{{Case studies - CRI~\eqref{eq:minimizationoftheratecriterion} [right] of the different system state estimation schemes for some topologies [left].}} 
			\label{fig:sim_convrates}
		\vspace{-10pt}
	\end{figure*}
	
Figure~\ref{fig:sim_convratesa} considers a \emph{small-world} network composed by three complete graphs having 6, 7 and 9 nodes interconnected by four edges. In this case, it is $\varsigma_{\normlap} < 1$ and the convergence performance of the schemes $\Sigma_\eta$ and $\Sigma_\rho$ improves when the regularization parameters are selected in $\widecheck{\mathcal{Q}}_\bullet$ ($r_{{\eta}^{\ast}} = 0.937, r_{\rho^{  \ast}} = 0.941$) rather then in $\mathcal{Q}_\bullet$ ($r_{\eta^{\ast}} = r_{\rho^{  \ast}} = r_{0} = 0.957$). In particular, the  corresponding CRIs approximate $r_{\Q^\ast} = 0.936$, while $r_{\epsilon^\ast} = 0.943$.
Also for the circulant graph $C_{36}(1,2)$ in Figure~\ref{fig:sim_convratesb}, it results $\varsigma_{\normlap} < 1$. Analogously to the previous case, it occurs that  $r_{\eta^{\ast}} = r_{\rho^{\ast}} = r_{\epsilon^{\ast}}  = 0.953 = r_{\Q^\ast}$ with $\eta^{\ast} \in \widecheck{\mathcal{Q}}_\eta$ and $\rho^{\ast} \in \widecheck{\mathcal{Q}}_\rho$, while $r_{\eta^{\ast}} = r_{\rho^{\ast}} = r_0 = 0.962$ with $\eta^{\ast} \in \mathcal{Q}_\eta$ and $\rho^{\ast} \in \mathcal{Q}_\rho$; the high topological symmetry implies that the optimal value corresponds to the same $q_i$ for all nodes.
Note also that with $C_{36}(1,2)$ also $C_{\Gmc}<H_{\Gmc}$ holds true. 

These two case studies confirm the advantages of the GRDS approach, through both the interpretation of the solutions of~\cite{fabris2019distributed,fabris2020proximal} in the extended domains and the potentialities entailed by the matrix $\Q$ itself.
	
The friendship graph in Figure~\ref{fig:sim_convratesc}, composed of $n=19$ nodes and constructed by joining $(n-1)/2$ cycle graphs $C_3(1,1)$ with a common vertex, is structurally characterized by $\varsigma_{\normlap} = 1$ (this property holds for any $n$-nodes friendship graph). In this case, 
accounting both for ${\mathcal{Q}}_\bullet$ and $\widecheck{\mathcal{Q}}_\bullet$, it is $r_{\eta^\ast} = r_{\rho^\ast} = 0.5= r_{\Q^\ast}$, namely the schemes $\Sigma_\eta$ and $\Sigma_\rho$ turn out to be optimal realizations of the generalized scheme $\Sigma_\Q$, independently on the parameter domain. Conversely, $\Sigma_\epsilon$ provides a high CRI ($r_{\epsilon^\ast}=0.9 > r_{\Q^\ast}$).
Finally, Figure~\ref{fig:sim_convratesd} shows the behavior of a Ramanujan topology having $\varsigma_{\normlap} > 1$. 
In detail, this is a regular graph consisting of $n=16$ nodes  with common degree $c =3$, and  satisfying $\max_{i\in \Imc_{1}} \{c \; |1-\lambda_{i}^{\normlap}| \} \leq 2\sqrt{c-1}$.
Here, the convergence performance of the schemes 
in~\cite{fabris2019distributed,fabris2020proximal,olfati2007consensus} does not improve when accounting for their GRDS interpretation. 
However, the optimization of the $q_i, i \in \mathcal{I}$, allows to find a regularization matrix $\Q^{*}$ corresponding to a lower CRI with respect to the other  selections ($r_{\Q^{*}} = 0.787 < r_{\eta^{*}} = r_{\rho^{*}} = r_\epsilon = 0.799$).

Notably, the index $r_{\Q^\ast}$ shows a transient behavior in Figure~\ref{fig:sim_convrates}(a) and (d), since the initialization in such cases is far from the optimal value; conversely, in Figure~\ref{fig:sim_convrates}(b)-(c), the  parameter selection is initialized already in correspondence to the optimal CRI and the applied perturbations are thus concentrated respectively above or on that value.

Indeed, the performance of the scheme $\Sigma_\Q$ is depending on the entries of matrix $\Q$, whose optimal selection is not straightforward and can be computed in closed form only for simple networks\footnote{E.g.: the line graph with $n=3$ and $d_1 = 2$, $d_2=d_3=1$ has $\varsigma_{\normlap}>1$; it can be proven that an optimal parameters selection is  $q^{\ast}_{1}=1-2q^{\ast}$, $q^{\ast}_{2}=q^{\ast}_{3}=q^{\ast}$, $q^{\ast} \in (0,1/2)$ arbitrarily small, leading to $r_{\Q}^{\ast}=q^{\ast}$.}. 
Nonetheless, and most importantly, the proposed GRDS realizations always allow to improve the convergence rate for the topologies characterized by $\varsigma_{\normlap} < 1$ and perform equally well in the other case.

\section{Concluding Remarks}\label{sec:conclusions}
	
This paper addresses the SERM problem in a regularized LS minimization framework proposing a GRDS approach that encompasses the state-of-the-art solutions~\cite{fabris2019distributed,fabris2020proximal,olfati2007consensus} via the introduction of the regularization matrix $\Q$.
A (sufficient) condition on the $\Q$ selection is stated ensuring the convergence  of the system state estimation towards the centralized solution.   
Numerical results highlight the strengths of the novel GRDS framework with respect to the convergence performance, especially in the case of {SNs} modeled by graphs having $\varsigma_{\normlap}<1$.
	Possible future research directions involve, for instance, the introduction of some  weights $w_{ij}\geq0$ in the cost function~\eqref{eq:cost_function}: these may represent the reliability on the communication link between the $i$-th and the $j$-th node or the trust/confidence level on the measurement $\tilde{x}_{ij}$. In the former case, it is reasonable to assume $w_{ij} = w_{ji}, \forall e_{ij} \in \Emc$ and the work outcome are still valid with minor modifications, including the use of weighted Laplacian matrix. In the latter case, it will be necessary to account for directed graphs.
	
	\bibliographystyle{IEEEtran}
\bibliography{new_bib}

\end{document}

%% file: shortcuts_math.tex
\newcommand{\Rset}[1]{\mathbb{R}^{#1}} 

\def \u {\mathbf{u}}

\def \x {\mathbf{x}}

\def \A {\mathbf{A}}

\def \D {\mathbf{D}} 
 
\def \F {\mathbf{F}}

\def \I {\mathbf{I}}

\def \L {\mathbf{L}} 
\def \M {\mathbf{M}} 
\def \N {\mathbf{N}}

\def \Q {\mathbf{Q}}

\def\Emc{\mathcal{E}}

\def\Gmc{\mathcal{G}}

\def\Imc{\mathcal{I}}

\def\Nmc{\mathcal{N}}

\def\Vmc{\mathcal{V}}

\def \ones			{{\mathds{1}}} 

\DeclareMathOperator{\diag}{diag} 

\newcommand{\eye}[1]{\mathbf{I}_{#1}} 

\def \nonodes 	{n} 
\def \lap		 		{\mathbf{L}} 
\def \normlap		 	{\boldsymbol{\mathcal{L}}} 

\newcommand{\rate}[1]		{\mathrm{r}_{#1}} 

%% file: main_LCSS_final.bbl
\begin{thebibliography}{10}
\providecommand{\url}[1]{#1}
\csname url@samestyle\endcsname
\providecommand{\newblock}{\relax}
\providecommand{\bibinfo}[2]{#2}
\providecommand{\BIBentrySTDinterwordspacing}{\spaceskip=0pt\relax}
\providecommand{\BIBentryALTinterwordstretchfactor}{4}
\providecommand{\BIBentryALTinterwordspacing}{\spaceskip=\fontdimen2\font plus
\BIBentryALTinterwordstretchfactor\fontdimen3\font minus
  \fontdimen4\font\relax}
\providecommand{\BIBforeignlanguage}[2]{{%
\expandafter\ifx\csname l@#1\endcsname\relax
\typeout{** WARNING: IEEEtran.bst: No hyphenation pattern has been}%
\typeout{** loaded for the language `#1'. Using the pattern for}%
\typeout{** the default language instead.}%
\else
\language=\csname l@#1\endcsname
\fi
#2}}
\providecommand{\BIBdecl}{\relax}
\BIBdecl

\bibitem{guastella2020cooperative}
D.~A. Guastella, V.~Campss, and M.-P. Gleizes, ``A cooperative multi-agent
  system for crowd sensing based estimation in smart cities,'' \emph{IEEE
  Access}, vol.~8, pp. 183\,051--183\,070, 2020.

\bibitem{lissandrini2019cooperative}
N.~Lissandrini, G.~Michieletto, R.~Antonello, M.~Galvan, A.~Franco, and
  A.~Cenedese, ``Cooperative optimization of uavs formation visual tracking,''
  \emph{Robotics}, vol.~8, no.~3, p.~52, 2019.

\bibitem{barooah2007estimation}
P.~Barooah and J.~P. Hespanha, ``Estimation on graphs from relative
  measurements,'' \emph{IEEE Control Sys. Mag.}, vol.~27, no.~4, pp. 57--74,
  2007.

\bibitem{xiong2017cooperative}
Y.~Xiong, N.~Wu, Y.~Shen, and M.~Z. Win, ``{\textcolor{black}{Cooperative
  network synchronization: Asymptotic analysis}},'' \emph{IEEE Trans. Signal
  Process.}, vol.~66, no.~3, pp. 757--772, 2017.

\bibitem{simonjan2019decentralized}
J.~Simonjan and B.~Rinner, ``{\textcolor{black}{Decentralized and
  resource-efficient self-calibration of visual sensor networks}},'' \emph{Ad
  Hoc Networks}, vol.~88, pp. 112--128, 2019.

\bibitem{dehghanpour2018survey}
K.~Dehghanpour, Z.~Wang, J.~Wang, Y.~Yuan, and F.~Bu, ``{\textcolor{black}{A
  survey on state estimation techniques and challenges in smart distribution
  systems}},'' \emph{IEEE Trans. Smart Grid}, vol.~10, no.~2, pp. 2312--2322,
  2018.

\bibitem{trimpe2014event}
S.~Trimpe and R.~D'Andrea, ``{\textcolor{black}{Event-based state estimation
  with variance-based triggering}},'' \emph{IEEE Trans. Autom. Control},
  vol.~59, no.~12, pp. 3266--3281, 2014.

\bibitem{farina2016distributed}
M.~Farina and R.~Carli, ``{\textcolor{black}{Distributed state estimation for
  independent linear systems with relative and absolute measurements}},'' in
  \emph{American Control Conf.}\hskip 1em plus 0.5em minus 0.4em\relax IEEE,
  2016, pp. 2029--2034.

\bibitem{aster2018parameter}
R.~C. Aster, B.~Borchers, and C.~H. Thurber, \emph{Parameter estimation and
  inverse problems}.\hskip 1em plus 0.5em minus 0.4em\relax Elsevier, 2018.

\bibitem{ravazzi2018distributed}
C.~Ravazzi, N.~P. Chan, and P.~Frasca, ``Distributed estimation from relative
  measurements of heterogeneous and uncertain quality,'' \emph{IEEE Trans.
  Signal Inf. Process. Netw}, vol.~5, no.~2, pp. 203--217, 2018.

\bibitem{shi2020bias}
M.~Shi, C.~De~Persis, P.~Tesi, and N.~Monshizadeh, ``{\textcolor{black}{Bias
  estimation in sensor networks}},'' \emph{IEEE Control Netw. Syst}, vol.~7,
  no.~3, pp. 1534--1546, 2020.

\bibitem{fabris2019distributed}
M.~Fabris, G.~Michieletto, and A.~Cenedese, ``On the distributed estimation
  from relative measurements: a graph-based convergence analysis,'' in
  \emph{European Control Conf.}\hskip 1em plus 0.5em minus 0.4em\relax IEEE,
  2019, pp. 1550--1555.

\bibitem{fabris2020proximal}
------, ``A proximal point approach for distributed system state estimation,''
  \emph{IFAC-PapersOnLine}, vol.~53, no.~2, pp. 2702--2707, 2020.

\bibitem{olfati2007consensus}
R.~Olfati-Saber, J.~A. Fax, and R.~M. Murray, ``Consensus and cooperation in
  networked multi-agent systems,'' \emph{Proc. IEEE}, vol.~95, no.~1, pp.
  215--233, 2007.

\bibitem{chung1997spectral}
F.~R. Chung and F.~C. Graham, \emph{Spectral graph theory}.\hskip 1em plus
  0.5em minus 0.4em\relax American Mathematical Soc., 1997, no.~92.

\bibitem{rojo2013new}
O.~Rojo and R.~L. Soto, ``A new upper bound on the largest normalized laplacian
  eigenvalue,'' \emph{Oper. Matrices}, vol.~7, pp. 323--332, 2013.

\end{thebibliography}
